\documentclass[11pt,a4paper,UKenglish,numberwithinsect]{article}

\usepackage[utf8]{inputenc}
\usepackage{fontenc}

\usepackage[left=2.5cm,right=2.5cm,top=2.5cm,bottom=3cm]{geometry}
\usepackage[pdfborder={0 0 0}]{hyperref}

\parskip = \medskipamount

\usepackage{amsmath,amsthm,amssymb}
\usepackage{mathtools}
\theoremstyle{theorem}
\newtheorem{theorem}{Theorem}[section]

\newtheorem{lemma}[theorem]{Lemma}

\theoremstyle{plain}
\newtheorem{example}[theorem]{Example}
\newtheorem{fact}[theorem]{Fact}

\newtheorem{definition}[theorem]{Definition}
\newtheorem{claim}[theorem]{Claim}

\usepackage{tikz}

\definecolor{uni}{HTML}{006374}
\definecolor{uniorange}{RGB}{236,116,4}
\definecolor{unired}{RGB}{228,32,50}
\definecolor{uniyellow}{RGB}{250,187,0}
\definecolor{unigreen}{RGB}{149,188,14}
\definecolor{uniblue}{RGB}{0,106,163}

\usepackage{booktabs}

\def\arena{\mathrm{arena}}
\def\colosseum{\mathrm{colosseum}}
\def\pit{\mathrm{pit}}

\usepackage{paralist}
\usepackage{listings}
\def\tw{\mathrm{tw}}
\def\Desc{\mathrm{Desc}}
\def\td{\mathrm{td}}
\def\pw{\mathrm{pw}}
\def\twq{\mathrm{tw}_q}
\def\dtw{\mathrm{dtw}}
\def\ball{\mathcal{B}}
\newcommand{\algop}[1]{\mathsf{#1}}

\title{Positive-Instance Driven Dynamic Programming\\ for Graph Searching}

\author{Max Bannach\kern1pt$^1$ \and Sebastian Berndt\kern1pt$^2$}

\date{\small%
  $^1$~Institute for Theoretical Computer Science,
  Universit\"at zu L\"ubeck,
  L\"ubeck, Germany \\
  \texttt{bannach@tcs.uni-luebeck.de}\\\vspace{2ex}
  $^2$~Department of Computer Science, Kiel University, Kiel, Germany\\
  \texttt{seb@informatik.uni-kiel.de}
}

\begin{document}

\maketitle 

\begin{abstract}
  \noindent Research on the similarity of a graph to being a tree~--~called the \emph{treewidth}
  of the graph~--~has seen an enormous rise within the last decade, but a
  practically fast algorithm for this task has been discovered only recently by
  Tamaki (ESA 2017). It is based on dynamic programming and makes use of the
  fact that the number of positive subinstances is typically substantially
  smaller than the number of all subinstances. Algorithms producing only such
  subinstances are called \emph{positive-instance driven} (PID). We give an
  alternative and intuitive view on this algorithm from the perspective of
  the corresponding configuration graphs in certain two-player games. This
  allows us to develop 
  PID-algorithms for a wide range of important graph parameters such
  as treewidth, pathwidth, and treedepth. We analyse the worst case behaviour of the
  approach on some well-known graph classes and perform an
  experimental evaluation on real world and random graphs.
\end{abstract}
\section{Introduction}
Treewidth, a concept to measure the similarity of a graph to being a tree, is arguably
one of the most used tools in modern combinatorial
optimization. It is a cornerstone of parameterized
algorithms~\cite{cygan:2015fr} and its success has led to its integration into
many different fields: For instance, treewidth and its close relatives treedepth
and pathwidth have been theoretically studied in the context of machine
learning~\cite{berg2014learning,darwiche2003differential,elidan2008learning},
model-checking~\cite{BannachB18,KneisLR11},
\textsc{SAT}-solving~\cite{BjesseKDSZ03,FichteHWZ18,HabetPT09},
\textsc{QBF}-solving~\cite{CharwatW16,EibenGO2018},
\textsc{CSP}-solving~\cite{KarakashianWC13,KosterHK02}, or
\textsc{ILPs}~\cite{EisenbrandHK18,GanianOR17,GanianO18,KouteckyLO18,Szeider03}.
Some of these results
(e.\,g.~\cite{BannachB18,BjesseKDSZ03,CharwatW16,FichteHWZ18,HabetPT09,KarakashianWC13,KneisLR11,KosterHK02})
show quite promising experimental results giving hope that the theoretical
results  lead to actual practical improvements. 

To utilize the treewidth for this task, we have to be able to compute it
quickly. More crucially, most algorithms also need a witness for this fact in
form of a tree-decomposition. In theory we have a beautiful algorithm for this
task~\cite{Bodlaender96}, which is unfortunately known to \emph{not} work in
practice due to huge constants~\cite{Roehrig98}. We may argue that, instead, a heuristic is
sufficient, as the attached solver will work correctly independently of the
actual treewidth~--~and the heuristic may produce a decomposition of ``small
enough'' width. However, 
even a small error, something as ``off by 5,'' may put the parameter
to a computationally intractable range, as the dependency on the
treewidth is usually at least exponential. It is therefore a very
natural and important task to build practical fast algorithms to
determine parameters as the treewidth or treedepth exactly.

To tackle this problem, the fpt-community came up with an implementation
challenge: the PACE~\cite{PaceIpec16,PaceIpec17}. Besides many, one very
important result of the challenge was a new combinatorial algorithm due to Hisao
Tamaki, which computes the treewidth of an input graph exactly and astonishingly
fast on a wide range of instances. An implementation of this algorithm by Tamaki
himself~\cite{Tamaki16} won the corresponding track in the PACE challenge in
2016~\cite{PaceIpec16} and an alternative implementation due to Larisch and
Salfelder~\cite{LarischS17} won in 2017~\cite{PaceIpec17}.
The algorithm is based on a dynamic program by Arnborg et
al.~\cite{ArnborgCP1987} for computing tree decompositions. This
algorithm has a game theoretic characterisation that we will utilities
in order to apply Tamaki's approach to a broader range of problems.
It should be noted, however, that Tamaki has improved his algorithm
for the second iteration of the PACE by applying his framework to the
algorithm by Bouchitt\'{e} and
Todinca~\cite{BouchitteT02,Tamaki2017}. This algorithm has a game
theoretic characterisation as well~\cite{FominK10}, but it is unclear
how this algorithm can be generalized to other parameters.
Therefore, we focus on Tamaki's first algorithm and analyze it both, from a theoretical and
a practical perspective. Furthermore, we will extend the algorithm to further graph parameters,
which is surprisingly easy due to the new game-theoretic
representation. In detail, our contributions are the following:

\begin{description}
\item[Contribution I: A simple description of Tamaki's first algorithm.]~\newline We
describe Tamaki's algorithm based on a well-known graph
searching game for treewidth. This provides a nice link to known theory and
allows us to analyze the algorithm in depth.

\item[Contribution II. Extending Tamaki's algorithm to other
  parameters.]~\newline The game theoretic point-of-view allows us to extend the
algorithm to various other parameters that can be defined in terms of similar
games~--~including pathwidth, treedepth.

\item[Contribution III: Experimental and theoretical analysis.]~\newline
We provide, for the first time, theoretical bounds on the runtime of the algorithm on
certain graph classes. Furthermore, we
count the number of subinstances generated by the algorithm on
various random and named graphs.
\end{description}

\section{Graph Searching}
A \emph{tree decomposition} of a graph $G=(V,E)$ is a tuple
$(T,\iota)$ consisting of a rooted tree $T$ and a mapping $\iota$ from
nodes of $T$ to sets of vertices of $G$ (called \emph{bags})
such that (1) for all $v\in V$ the set $\{\,x\mid v\in\iota(x)\,\}$ is
nonempty and connected in $T$, and (2) for every edge $\{v,w\}\in E$
there is a node $m$ in $T$ with $\{v,w\}\subseteq\iota(m)$. The
\emph{width} of a tree decomposition is the maximum size of one of its
bags minus one, its \emph{depth} is the maximum of the width and the depth of~$T$. The \emph{treewidth} of~$G$, denoted by $\tw(G)$,
is the minimum width any tree decomposition of $G$ must have. If $T$
is a path we call $(T,\iota)$ a \emph{path decomposition}; if for all
nodes $x,y$ of $T$ we have $\iota(x)\subsetneq\iota(y)$ whenever $y$
is a descendent of $x$ we call $(T,\iota)$ a \emph{treedepth
  decomposition}; and if on any path from the root to a leaf there are
at most $q$ nodes with more then one children we call $(T,\iota)$ a
\emph{$q$-branched tree decomposition}. Analogous to the treewidth, we
define the \emph{pathwidth} and
\emph{$q$-branched-treewidth} of $G$, denoted by $\pw(G)$
and $\twq(G)$, respectively. The \emph{treedepth} $\td(G)$ is the minimum depth any treedepth decomposition must
have. The various parameters are illustrated in Figure~\ref{figure:tw}.

\begin{figure}
  \begin{center}\includegraphics{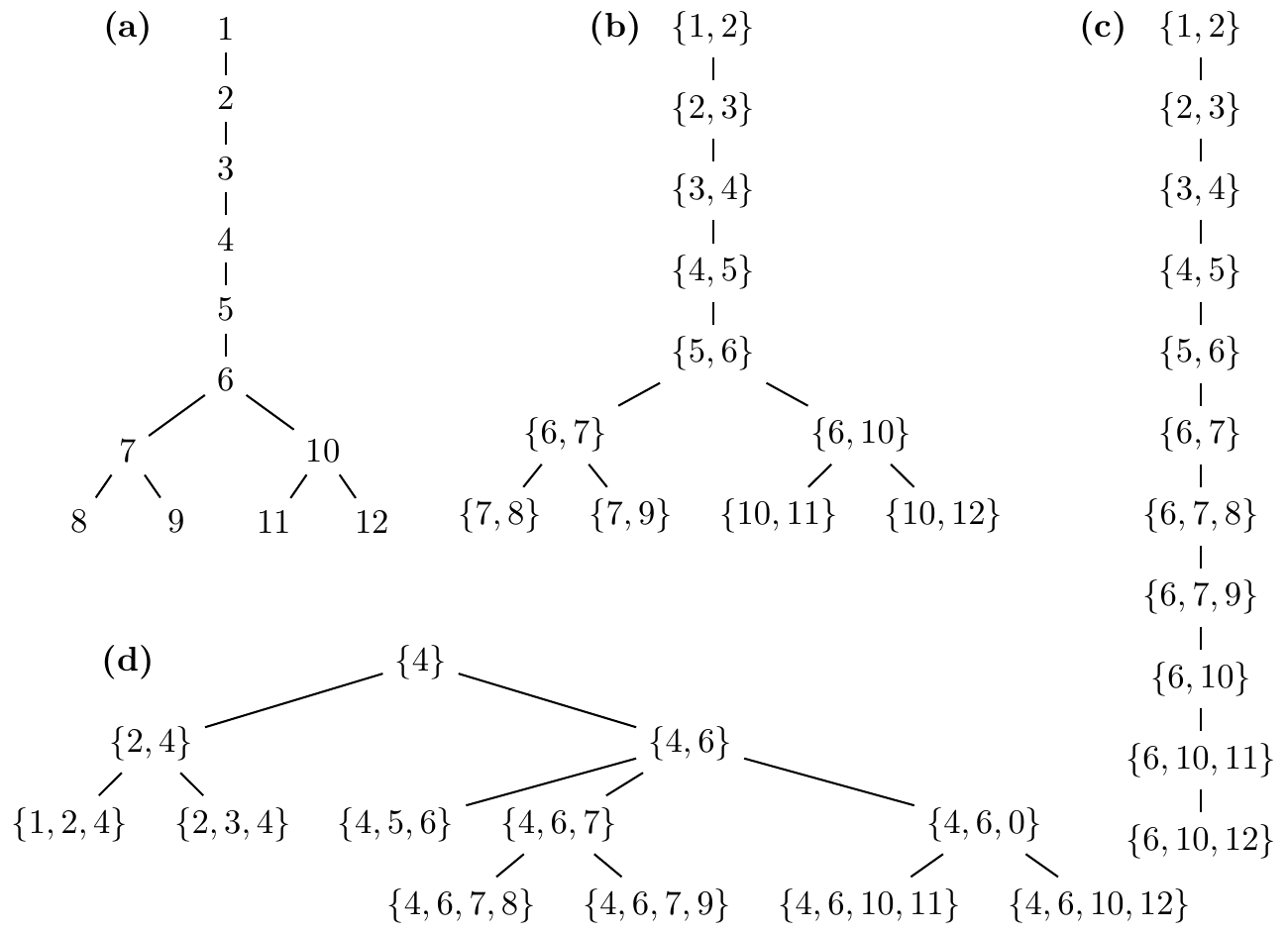}\end{center}
  \caption{Various tree decompositions of an undirected graph $G=(V,E)$ shown at
    {\bf(a)}. The decompositions justify {\bf(b)} $\tw(G)\leq 1$, {\bf(c)}
    $\pw(G)\leq 2$, and {\bf(d)} $\td(G)\leq 3$. With respect to
    $q$-branched treewidth the decompositions also justify {\bf(b)}
    $\tw_2(G)\leq 1$ and {\bf(c)} $\tw_0(G)\leq 2$.}
  \label{figure:tw}
\end{figure}

Another important variant of this parameter is
\emph{dependency-treewidth}, which is used primarily in the context of
quantified Boolean formulas~\cite{EibenGO2018}. For a graph $G=(V,E)$ and a
partial order $\lessdot$ of $V$ the
dependency-treewidth $\dtw(G)$ is the minimum width any
tree-decomposition $(T,\iota)$ with the following property must have:
Consider the natural partial order $\leq_T$ that $T$ induces on its
nodes, where the root is the smallest elements and the leaves form the
maximal elements; define for any $v\in V$ the node $F_v(T)$ that is
the $\leq_T$-minimal node $t$ with $v\in\iota(t)$ (which is well
defined); then define a partial order $<_{\mathcal{T}}$ on $V$ such
that $u<_{\mathcal{T}} v\Longleftrightarrow F_u(T)\leq_T
F_v(T)$; finally for all $u,v\in V$ it must hold that
$F_u(T)<_{T} F_v(T)$ implies that that $u\lessdot v$ does not hold.

We study classical graph searching in a general setting proposed by
Fomin, Fraigniaud, and Nisse~\cite{FominFN09}. The input is an
undirected graph $G=(V,E)$ and a number $k\in\mathbb{N}$, and the question
is whether a team of $k$ searchers can catch an \emph{invisible} fugitive on $G$ by the
following set of rules: At the beginning, the fugitive is placed at
a vertex of her choice and at any time, she knows the position of the
searchers. In every turn she may move with \emph{unlimited speed}
along edges of the graph, but may never cross a vertex occupied by
a searcher. This implies that the fugitive does not occupy a
single vertex but rather a subgraph, which is
separated from the rest of the graph by the searchers. The vertices of
this subgraph are called
\emph{contaminated} and at the start of the game all vertices are
contaminated.
The searchers, trying to catch the fugitive, can perform one of the
following operations during their turn:
\begin{compactenum}
  \item \emph{place} a searcher on a contaminated vertex;
  \item \emph{remove} a searcher from a vertex;
  \item \emph{reveal} the current position of the fugitive.
\end{compactenum}  
When a searcher is placed on a contaminated vertex it becomes
\emph{clean}. When a searcher is removed from a vertex $v$, the vertex
may become \emph{recontaminated} if there is a contaminated vertex
adjacent to $v$. The searchers win the game if they manage to clean
all vertices, i.\,e., if they catch the fugitive; the fugitive wins
if, at any point, a recontamination occurs, or if she can escape
infinitely long. Note that this implies that the
searchers have to catch the fugitive in a \emph{monotone} way.
A priori one could assume that the later condition gives the
fugitive an advantage (recontamination could be
necessary for the cleaning strategy), however, a
crucial result in graph searching is that ``recontamination does not
help'' in all variants of the game that we
consider~\cite{BienstockS1991,GiannopoulouHT12,LaPaugh1993,SeymourT1993,MazoitN08}.

\subsection{Entering the Arena and the Colosseum}
Our primary goal is to determine whether the searchers have a winning
strategy. A folklore algorithm for this task is to construct an
alternating graph $\arena(G,k)=((V_s\cup V_f),E_{\mathrm{ar}})$ that contains for each
position of the searchers ($S\subseteq V$ with $|S|\leq k$) and
each position of the fugitive ($f\in V$) two copies of the vertex 
$(S,f)$, one in $V_s$ and one in $V_f$ (see e.\,g.~Section~7.4~in~\cite{cygan:2015fr}). Vertices in $V_s$ correspond
to a configuration in which the searchers do the next move (they are
existential) and vertices in $V_f$ correspond to fugitive moves
(they are universal). The edges $E_{\mathrm{ar}}$ are constructed according to the possible moves.
Clearly, our task is now reduced to the question
whether there is an alternating path from a start configuration to
some configuration in which the fugitive is caught. Since alternating
paths can be computed in linear time (see e.\,g.,~Section~3.4~in~\cite{Immerman1999}), we immediately obtain an
$O(n^{k+1})$ algorithm.

Modeling a configuration of the game as tuple $(S,f)$ comes, however, with a
major drawback: The size of the arena does directly depend on $n$ and $k$ and
does \emph{not} depend on some further structure of the input. For instance, the
arena of a path of length $n$ and any other graph on $n$ vertices will have the
same size for any fixed value $k$. As the major goal of parameterized complexity
is the understanding of structural parameters beyond the input size $n$, such a
fixed-size approach is usually not practically feasible. In contrast, we will
define the configuration graph $\colosseum(G,k)$, which might be larger then
$\arena(G,k)$ in general, but is also ``prettier'' in the sense that it adapts
to the input structure of the graph. Moreover, the resulting algorithms are
\emph{self-adapting} in the sense that it needs no knowledge about this special
structure to make use of it (in constrast to other parameterized algorithms,
where the parameter describing this structure needs to be given explicitly).

\subsection{Simplifying the Game}

Our definition is based upon a similar formulation by Fomin et
al.~\cite{FominFN09}, but we
simplify the game to make it more accessible to our techniques. 
First of all, we restrict the fugitive in the following
sense. Since she is invisible to the searchers and travels with
unlimited speed, there is no need for her to take regular
actions. Instead, the only moment when she is actually active is when
the searchers perform a reveal. If $C$ is the set of contaminated
vertices, consisting of the induced components $C_1,\dots,C_{\ell}$, a
reveal will uncover the component in which the fugitive hides and,
as a result, reduce $C$ to $C_i$ for some $1\leq i\leq\ell$. The only
task of the fugitive is, thus, to answer a reveal with such a number
$i$. We call the whole process of the searcher performing a reveal, the
fugitive answering it, and finally of reducing $C$ to $C_i$ a \emph{reveal-move}.

We will also restrict the searchers by the concept of \emph{implicit searcher removal.} Let
$S\subseteq V(G)$ be the vertices currently occupied by the searchers, and
let $C\subseteq V(G)$ be the set of contaminated vertices. We
call a vertex $v\in S$ 
\emph{covered} if every path between $v$ and $C$ contains a vertex 
$w\in S$ with $w\neq v$.

\begin{lemma}\label{lemma:covered}
A covered searcher can be removed safely.
\end{lemma}
\begin{proof}
As we have $N(v)\cap C=\emptyset$, the removal of $v$ will not
increase the contaminated area. Furthermore, at no later point of the
game $v$ can be recontaminated, unless a neighbor of $v$ gets
recontaminated as well (in which case the game would already be lost
for the searchers).
\end{proof}
\begin{lemma}\label{lemma:onlyCovered}
Only covered searchers can be removed safely. 
\end{lemma}
\begin{proof}
Since for any other vertex $w\in S$ we have $N(w)\cap C\neq\emptyset$,
the removal of $w$ would recontaminate $w$ and, hence, would result in
a defeat of the searchers.
\end{proof}
Both lemmas together imply that the searchers never have to decide to
remove a searcher, but rather do it \emph{implicitly}. We thus
restrict the possible moves of the searchers to a combined move of
placing a searcher and \emph{immediately} removing the searchers from all covered
vertices. We call this a \emph{fly-move}. Observe that the sequence of
original moves mimicked by a fly-move does not contain a reveal and,
thus, may be performed independently of any action of the fugitive.

We are now ready to define the colosseum. We could, as for the arena,
define it as an alternating graph. However, as the searcher is the
only player that performs actions in our simplified game, we find it
more natural to express this game as \emph{edge-alternating
  graph}~--~a generalization of alternating graphs. An
edge-alternating graph is a triple $H=(V,E,A)$ consisting of a
\emph{vertex set} $V$, an existential edge relation $E\subseteq
V\times V$, and an universal edge relation $A\subseteq V\times V$.
We define the neighborhood of a vertex $v$ as $N_{\exists}(v)=\{\,w\mid(v,w)\in E\,\}$,
$N_{\forall}(v)=\{\,w\mid(v,w)\in A\,\}$, and $N_H(v)=N_{\exists}(v)\cup N_{\forall}(v)$.
An \emph{edge-alternating} $s$-$t$-path is a set $P\subseteq V$ such that
(1)~$s,t\in P$ and (2)~for all $v\in P$ with $v\neq t$ we have either \(
N_{\exists}(v)\cap P\neq\emptyset \) or \( \emptyset\neq
N_{\forall}(v)\subseteq P \) or both. We write $s\prec t$ if such a
path exists and define $\ball(Q)=\{\,v\mid v\in Q\vee(\exists w\in
Q\colon v\prec w)\,\}$ for $Q\subseteq V$ as the set of vertices on edge-alternating
paths leading to $Q$. We say that an edge-alternating $s$-$t$-path $P$
is \emph{$q$-branched}, if (i) $H$ is
acyclic and (ii) every (classical) directed path $\pi$ from $s$ to $t$
in $H$ with $\pi\subseteq P$ uses at most $q$ universal edges.

For an undirected graph $G=(V,E)$ and a number $k\in\mathbb{N}$ we now define
the $\colosseum(G,k)$ to be the edge-alternating graph $H$ with vertex set
$V(H)=\{\,C\mid \text{$\emptyset\neq C\subseteq V$ and $|N_G(C)|\leq k$}\,\}$ and
the following edge sets: for all pairs $C,C'\in V(H)$ there is an edge
$e=(C,C')\in E(H)$ if, and only if, $C\setminus\{v\}=C'$ for some
$v\in C$ and $|N_G(C)| < k$; furthermore, for all
$C\in V(H)$ with at least two components $C_1,\dots,C_{\ell}$ we have edges
$(C,C_i)\in A(H)$. The \emph{start configuration} of the game is the vertex
$C=V$, that is, all vertices are contaminated. We define $Q=\{\, \{v\}\subseteq
V \colon |N_G(\{v\})|<k\,\}$ to be the set of \emph{winning configurations}, as at least
one searcher is available to catch the fugitive. Therefore, the searchers have a
winning strategy if, and only if, $V\in \ball(Q)$ and we will therefore refer to
$\ball(Q)$ as the \emph{winning region}. Observe that the colosseum is
acyclic (that is, the digraph $(V,E\cup A)$ is acyclic) as we have for every edge $(C, C')$ that $|C|>|C'|$, and
observe further that $Q$ is a subset of the sinks of $H$. Hence, we can 
test if $V\in \ball(Q)$ in time $O(|\colosseum(G,k))|)$. Finally, note that
the size of $\colosseum(G,k)$ may be of order $2^n$ rather than $n^{k+1}$,
giving us a slightly worse overall runtime. 

The reader that is familiar with graph searching or with exact
algorithms for treewidth will probably notice the similarity of the
colosseum and an exact ``Robertson--Seymour fashioned'' algorithm as
sketched in Listing~\ref{alg:ggs}.

\begin{lstlisting}[float,caption={To get some intuition behind the
    colosseum, consider the following procedure. It is assumed that an input graph
    $G=(V,E)$ and a target number $k\in\mathbb{N}$ is globally
    available in memory. The procedure, when called with
    parameters $S=\emptyset$ and $C=V$, will determine whether $k$ searcher can
    catch the fugitive in the search game. Hereby, the set $S$ is
    always the current position of the searchers and $C$ is the
    contaminated area. We maintain the invariant $N(C)\subseteq S$, as the
searchers would lose otherwise due to recontamination. Observe that from any
    configuration $(S,C)$ the procedure will, without branching, move
    to $(N(C),C)$. These are exactly the configurations that are
  present in the colosseum. In fact, the colosseum is essentially the
  configuration graph of this procedure if it is used with memoization.},
  numbers=left, label=alg:ggs]
procedure $\algop{generalGraphSearching}(S,C)$

  // end of recursion
  if $|S|>k$ then // we need too many searchers
    return false
  end
  if $C=\emptyset$ then // the searchers cleaned the graph
    return true
  end

  // implicit searcher removal
  for $v\in S$ do
    if $N(v)\cap C=\emptyset$ then
      $S\leftarrow S\setminus\{v\}$
      return $\algop{generalGraphSearching}(S,C)$
    end
  end

  // reveal-move
  $C_1,\dots,C_{\ell}\leftarrow\algop{connectedComponents}(G[C])$
  if $\ell > 1$ then
    return $\bigwedge_{i=1}^{\ell}\algop{generalGraphSearching}(S,C_i)$
  end

  // fly-move
  return $\bigvee_{v\in C}\algop{generalGraphSearching}(S\cup\{v\},C\setminus\{v\})$
end
\end{lstlisting}

\clearpage
\subsection{Fighting in the Pit}
Both algorithms introduced in the previous section run asymptotically
in the size of the generated configuration graph 
$|\arena(G,k)|$ or $|\colosseum(G,k)|$. Both of these graphs might be very
large, as the arena has fixed size of order $O(n^{k+1})$, while the colosseum may even have size $O(2^n)$. Additionally,
both graphs contain many unnecessary configurations, that is,
configurations that are not contained in the winning region of the searchers.
In the light of dynamic programming this is the same
as listing all possible configurations; and in the light of
positive-instance driven dynamic programming we would like to
list only the positive instances~--~which is exactly the winning
region in this context.

To realize this idea, we consider the \emph{pit} inside the colosseum,
which is the area where only true champions can survive~--~formally we
define $\pit(G,k)$ as the subgraph of $\colosseum(G,k)$ induced by
$\ball(Q)$, that is, as the induced subgraph on the winning
region. The key-insight is that $|\pit(G,k)|$ may be smaller
than $|\colosseum(G,k)|$ or even $|\arena(G,k)|$ on various graph
classes. Our primary goal for the next section will therefore be the
development of an algorithm that computes the pit in time
$O(|\pit(G,k)|^2)$.

\section{Computing the Pit}
Our aim for this section is to develop an algorithm that computes
$\pit(G,k)$. Of course, a simple way to do this is to
compute the whole colosseum and to extract the pit
afterwards. However, this will cost time $O(2^n)$ and is surely not
what we aim for. Our algorithm traverses
the colosseum ``backwards'' by starting at the set $Q$ of winning
configurations and by uncovering $\ball(Q)$ layer by layer. In order
to achieve this, we need to compute the predecessors of a
configuration~$C$. This is easy if $C$ was reached by a
fly-move as we can simply enumerate the $n$ possible predecessors. 
Reversing a reveal-move, that is, finding the
universal predecessors, is significantly more involved. A simple
approach is to test for every subset of already explored
configurations if we can ``glue'' them together~--~but this would
result in an even worse runtime of $2^{|\pit(G,k)|}$. Fortunately, we
can avoid this exponential blow-up as the colosseum has the
following useful property:
\begin{definition}[Universal Consistent]
  We say that an edge-alternating graph $H=(V,E,A)$ is \emph{universal consistent}
  with respect to a set $Q\subseteq V$
  if for all $v\in V\setminus Q$ with $v\in\ball(Q)$ and
  $N_{\forall}(v)=\{w_1,\dots,w_r\}$ we have (1)
  $N_{\forall}(v)\subseteq\ball(Q)$ and (2)~for every
  $I\subseteq\{w_1,\dots,w_r\}$ with $|I|\geq 2$ there is a vertex $v'\in V$ with
  $N_{\forall}(v')=I$ and $v'\in\ball(Q)$.
\end{definition}

Intuitively, this definition implies that for every vertex with high
universal-degree there is a set of vertices that we can arrange in a
tree-like fashion to realize the same adjacency relation. This allows us to glue only two configurations at
a time and, thus, removes the exponential dependency. An example of the
definition can be found in Example~\ref{example:universalconsitent}.

\begin{example}\label{example:universalconsitent}
  Consider the following three edge-alternating graphs, where a black
  edge is existential and the \textcolor{red!50!black}{red} edges are
  universal. The set $Q$ contains a single vertex that is
  highlighted. From left to right: the first graph is universal
  consistent; the second and third one are not. The second graph
  conflicts the condition that $v\in\ball(Q)$ implies
  $N_{\forall}(v)\subseteq\ball(Q)$, as the vertex on the very left is
  contained in $\ball(Q)$ by the top path, while its universal
  neighbor on the bottom path is not contained in $\ball(Q)$. The
  third graph conflicts the condition that  $N_{\forall}(v)=\{w_1,\dots,w_r\}$ implies that for every
  $I\subseteq\{w_1,\dots,w_r\}$ with $|I|\geq 2$ there is a vertex $v'\in V$ with
  $N_{\forall}(v')=I$ and $v'\in\ball(Q)$ as witnessed by the vertex with
  three outgoing universal edges. 

  \vspace{2ex}
  \begin{center}
    \includegraphics{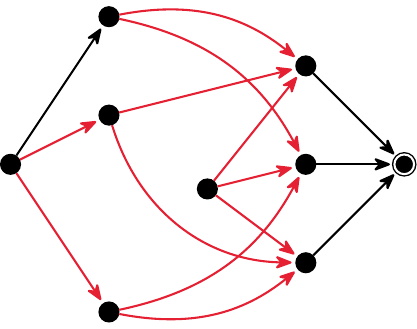}%
  \qquad%
    \includegraphics{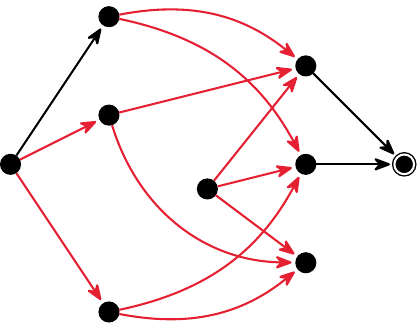}%
  \qquad%
    \includegraphics{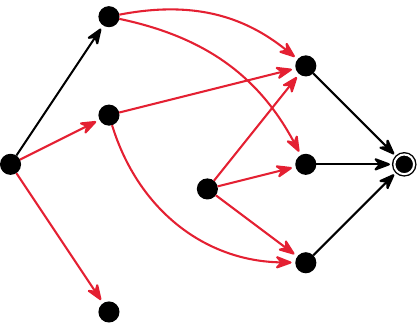}%
  \end{center}
\end{example}

\begin{lemma}\label{lemma:universal_consistent}
  For every graph $G=(V,E)$ and number $k\in \mathbb{N}$, the
  edge-alternating graph
  $\colosseum(G,k)$ is universal consistent.
\end{lemma}
\begin{proof}
  For the first property just observe that ``reveals do not harm'' in
  the sense that if the searchers can catch the fugitive without
  knowing where she hides, they certainly can do if they do know.
  
  For the second property consider any configuration $C\in V(H)$ that has universal edges to
  $C_1,\dots,C_{\ell}$. By definition we have $|N(C)|\leq k$ and
  $N(C_i)\subseteq N(C)$ for all $1\leq i\leq\ell$. Therefore we have
  for every $I\subseteq\{1,\dots,\ell\}$ and $C'=\cup_{i\in I}C_i$ that
  $N(C')\subseteq N(C)$ and $|N(C')|\leq k$ and, thus, $C'\in V(H)$.
\end{proof}
We are now ready to formulate the algorithm for computing the pit shown in Listing~\ref{alg:pid:alternating}. In essence, the
algorithm runs in three phases: first it computes the set $Q$ of
winning configurations; then the winning region $\ball(Q)$ (that
is, the vertices of $\pit(G,k)$); and finally, it
computes the edges of $\pit(G,k)$.

\begin{theorem}\label{theorem:pid}
  The algorithm $\mathsf{Discover}(G,k)$
  finishes in at most $O\big(|\ball(Q)|^2\cdot |V|^{2}\big)$ steps and
  correctly outputs $\pit(G,k)$.
\end{theorem}
\begin{proof}
    The algorithm is supposed to compute $Q$ in phase I, $\ball(Q)$ in
  phase II, and the edges of $\colosseum(G,k)[\ball(Q)]$ in phase
  III. First observe that $Q$ is correctly computed in phase I by the
  definition of $Q$.

  To show the correctness of the second phase we argue that the computed set
  $V(\pit(G,k))$ equals $\ball(Q)$. Let us refer to the set
  $V(\pit(G,k))$ during the computation as $K$ and 
  observe that this is exactly the set of vertices
  inserted into the queue. We first show $K\subseteq \ball(Q)$ by induction over the $i$th
  inserted vertex. The first vertex $C_1$ is in $\ball(Q)$ as $C_1\in
  Q$. Now consider $C_i$. As $C_i\in K$, it was either added in
  Line~\ref{discover:F} or Line~\ref{discover:R}. In the first case
  there was a vertex $\tilde C_i\in K$ such that $C_i=\tilde
  C_i\cup\{v\}$ for some $v\in N(\tilde C_i)$. By the induction
  hypothesis we have $\tilde C_i\in\ball(Q)$ and by the definition of
  the colosseum we have $(C_i,\tilde C_i)\in E(H)$ and, thus,
  $C_i\in\ball(Q)$. In the second case there where vertices $\tilde
  C_i$ and $\hat C_i$ with $\tilde C_i,\hat C_i\in K$ and $C_i=\tilde
  C_i\cup\hat C_i$. By the induction hypothesis we have again $\tilde
  C_i,\hat C_i\in\ball(Q)$. Let $t_1,\dots, t_{\ell}$ be the connected components of $\tilde C_i$ and
  $\hat C_i$. Since the colosseum $H$ is universal consistent with
  respect to $Q$ by
  Lemma~\ref{lemma:universal_consistent}, we have $t_1,\dots,
  t_{\ell}\in\ball(Q)$. By the definition of the colosseum we have
  $N_{\forall}(C_i)=t_1,\dots,t_{\ell}$ and, thus, $C_i\in\ball(Q)$.

  To see $\ball(Q)\subseteq K$ consider for a contradiction the vertices of $\ball(Q)$ in
  reversed topological order (recall that $H$ is acyclic) and let $C$
  be the first vertex in this order with $C\in\ball(Q)$ and $C\not\in K$.
  If $C\in Q$ we have $C\in K$ by phase I and are done, so assume
  otherwise. Since $C\in\ball(Q)$ we have either
  $N_{\exists}(C)\cap\ball(Q)\neq\emptyset$ or $\emptyset\neq N_{\forall}(C)\subseteq\ball(Q)$.
  In the first case there is a $\tilde C\in\ball(Q)$ with $(C,\tilde
  C)\in E(H)$. Therefore, $\tilde C$ precedes $C$ in the reversed
  topological order and, by the choice of $C$, we have $\tilde C\in
  K$. Therefore, at some point of the algorithm $\tilde C$ gets
  extracted from the queue and, in Line~\ref{discover:F}, would add
  $C$ to $K$, a contradiction.

  In the second case there are vertices $t_1,\dots,t_{\ell}$ with
  $N_{\forall}(C)=\{t_1,\dots,t_{\ell}\}$ and
  $t_1,\dots,t_{\ell}\in\ball(Q)$. By the choice of $C$, we have again
  $t_1,\dots,t_{\ell}\in K$. Since $H$ is universal consistent with
  respect to $Q$, we have for every $I\subseteq\{1,\dots,\ell\}$ that
  $\bigcup_{i\in I}t_i$ is contained in $\ball(Q)$. In particular, the
  vertices $t_1\cup t_2$, $t_3\cup t_4$, $\dots$, $t_{\ell-1}\cup t_{\ell}$ are contained in
  $\ball(Q)$, and these elements are added to $K$ whenever the $t_i$ are
  processed (for simplicity assume here that $\ell$ is a power of 2).
  Once these elements are processed, Line~\ref{discover:R}
  will also add their union, that is, vertices of the form $(t_1\cup
  t_2)\cup(t_3\cup t_4)$. In this way, the process will add vertices
  that correspond to increasing subgraphs of $G$ to $K$, resulting 
  ultimately in adding $\bigcup_{i=1}^{\ell}t_i=C$ into $K$, which is the
  contradiction we have been looking for.

  Finally, once the set $\ball(Q)$ is known, it is easy to compute the
  subgraph $\colosseum(G,k)[\ball(Q)]$, that is, to compute the edges
  of the subgraph induced by $\ball(Q)$. Phase III essentially
  iterates over all vertices and adds edges according to the
  definition of the colosseum.
  
  For the runtime, observe that the queue will contain exactly the set
  $\ball(Q)$ and, for every element extracted, we search through the
  current $K'\subseteq\ball(Q)$, which leads to the quadratic
  timebound of $|\ball(Q)|^2$. Furthermore, we have to compute the
  neighborhood of every extracted element, and we have to test whether two
  such configurations intersect~--~both can easily be achieved in time
  $O(|V|^2)$. Finally, in phase~III we have to compute connected
  components of the elements in $\ball(Q)$, but since this is possible
  in time $O(|V|+|E|)$ per element, it is clearly possible in time
  $|\ball(Q)|\cdot|V|^2$ for the whole graph.
\end{proof}

\begin{minipage}[t]{.51\linewidth}
  \begin{lstlisting}[caption={$\algop{Discover}(G,k)$},
    numbers=left, label=alg:pid:alternating]
$V(\pit(G,k)) := \emptyset$
$E(\pit(G,k)) := \emptyset$
$A(\pit(G,k)) := \emptyset$
initialize empty queue

// Phase I: compute $Q$
for $v\in V(G)$ do
  $\algop{insert}(\{v\}, k-1)$
end
  
// Phase II: compute $\ball(Q)=V(\pit(G,k))$
while queue not empty do 
  extract $C$ from queue
  // reverse fly-moves
  for $v\in N(C)$ do
    $\algop{insert}(C\cup\{v\}, k-1)$$\label{discover:F}$
  end
  // reverse reveal-moves
  for $C'\in V(\pit(G,k))$ with $C\cap C'=\emptyset$ do
    $\algop{insert}(C\cup C', k)$$\label{discover:R}$
  end  
end
  
// Phase III: compute $E$ and $A$
discoverEdges()

return $\big(V(\pit(G,k)), E(\pit(G,k)), A(\pit(G,k))\big)$
\end{lstlisting}
\end{minipage}\hspace{0.75cm}
\begin{minipage}[t]{.45\linewidth}
\begin{lstlisting}[caption={$\algop{insert}(C, t)$}, numbers=right]
if $C\not\in V(\pit(G,k))$ and $|N_G(C)|\leq t$ then
  add $C$ to $V(\pit(G,k))$
  insert $C$ into queue
end
\end{lstlisting}\vspace{5ex}
\begin{lstlisting}[caption={$\algop{discoverEdges}()$}, numbers=right]
for $C\in V(\pit(G,k))$ do

  // add fly-move edges
  for $v\in C$ do
    if $C\setminus\{v\}\in V(\pit(G,k))$ then
      add $(C, C\setminus\{v\})$ to $E(\pit(G,k))$
    end
  end

  // add reveal-move edges
  let $C_1,\dots,C_{\ell}$ be
    the connected components of $G[C]$
  if $C_1,\dots,C_{\ell}\in K$ then
    for $i=1$ to $\ell$ do
      add $(C, C_i)$ to $A(\pit(G,k))$
    end
  end
    
end
\end{lstlisting}
\end{minipage}

\section{Distance Queries in Edge-Alternating Graphs}\label{section:labeling}
In the previous section we have discussed how to compute the pit for a given graph
and a given value $k$. The computation of treewidth now boils down to a
reachability problem within this pit. But, intuitively, the pit should be able
to give us much more information. In the present section we formalize this
claim: We will show that we can compute shortest edge-alternating paths. To get
an intuition of ``distance'' in edge-alternating graphs think about such a graph
as in our game and consider some vertex $v$. There is always one active player
that may decide to take \emph{one} existential edge (a fly-move in our game), or
the player may decide to ask the opponent to make a move and, thus, has to
handle \emph{all} universal edges (a reveal-move in our game). From the point of
view of the active player, the distance is thus the \emph{minimum}
over the minimum of the
distances of the existential edges and the maximum of the universal edges.
\begin{definition}[Edge-Alternating Distance]
  Let $H=(V,E,A)$ be an edge-alternating graph with $v\in V$ and
  $Q\subseteq V$, let further $c_0\in\mathbb{N}$ be a constant and
  $\omega_E\colon E\rightarrow\mathbb{N}$ and $\omega_A\colon
  A\rightarrow\mathbb{N}$ be weight functions. The
  distance $d(v,Q)$ from $v$ to $Q$ is inductively defined as
  $d(v,Q)=c_0$ for $v\in Q$ and otherwise:
  \[
    d(v,Q)=
      \min\big(\,
      \min\limits_{w\in N_{\exists}(v)}(d(w,Q)+\omega_E(v,w)),\,
      \max\limits_{w\in N_{\forall}(v)}(d(w,Q)+\omega_A(v,w))\,
      \,\big).
  \]
\end{definition}

\begin{lemma}\label{lemma:distance}
  Given an acyclic edge-alternating graph $H=(V,E,A)$, weight
  functions $\omega_E\colon E\rightarrow\mathbb{N}$ and
  $\omega_A\colon A\rightarrow\mathbb{N}$, a source vertex $s\in V$,
  a subset of the sinks $Q$, and a constant $c_0\in\mathbb{N}$. The value $d(s,Q)$ can be
  computed in time $O(|V|+|E|+|A|)$ and a corresponding edge-alternating path can be computed in the same time.
\end{lemma}
\begin{proof}
  Since $H$ is acyclic we can compute a topological order of $V$ using the
  algorithm from~\cite{Kahn62}. We iterate over the vertices $v$ in reversed
  order and compute the distance as follows: if $v$ is a sink we either set
  $d(v,Q)=c_0$ or $d(v,Q)=\infty$, depending on whether we have $v\in Q$. If $v$
  is not a sink we have already computed $d(w,Q)$ for all $w\in N(v)$ and,
  hence, can compute $d(v,Q)$ by the formula of the definition. Since this
  algorithm has to consider every edge once, the whole algorithm runs in time
  $O(|V|+|E|+|A|)$. A path from $s$ to $Q$ of length $d(s,Q)$ can be found by
  backtracking the labels starting at $s$.
\end{proof}

\begin{theorem}\label{theorem:parameters}
  Given a graph $G=(V,E)$ and a number $k\in\mathbb{N}$, we can decide in time $O(|\pit(G,k+1)|^2\cdot|V|^2)$
  whether $G$ has $\{$ treewidth, pathwidth, treedepth,
  $q$-branched-treewidth, dependency-treewidth $\}$ at most $k$.
\end{theorem}
Before we got into the details, let us briefly sketch the general idea
of proving the theorem: All five problems have game theoretic characterizations in terms of
  the same search game with the same configuration set~\cite{BienstockS1991,FominFN09,GiannopoulouHT12}. More
  precisely, they condense to various distance questions within the
  colosseum by assigning appropriate weights to the edges. 
  \begin{compactdesc}
  \item[treewidth:] To solve treewidth, it is sufficient to find \emph{any}
    edge-alternating path from the vertex $C_s=V(G)$ to a vertex in
    $Q$. We can find a path by choosing
    $\omega_{E}$ and $\omega_{A}$ as $(x,y)\mapsto 0$, and by setting $c_0=0$.
  \item[pathwidth:] In the pathwidth game, the searchers are not allowed to
    perform any reveal~\cite{BienstockS1991}. Hence, universal edges cannot
    be used and we set $\omega_{A}$ to $(x,y)\mapsto\infty$. By setting
    $\omega_{E}$ to $(x,y)\mapsto 0$ and $c_0=0$, we again only need to find some path from $V(G)$ to
    $Q$ with weight less than $\infty$.
  \item[treedepth:] In the game for treedepth, the searchers are not allowed to
    remove a placed searcher again~\cite{GiannopoulouHT12}. Hence, the searchers
    can only use $k$ existential edges. Choosing $\omega_{E}$ as $(x,y)\mapsto 1$,
    $\omega_{A}$ as $(x,y)\mapsto 0$, and $c_0=1$ is sufficient. We
    have to search a path of weight at most $k$.
  \item[$q$-branched-treewidth:] For q-branched-treewidth we wish to
    use at most $q$ reveals~\cite{FominFN09}. By choosing $\omega_{E}$
    as $(x,y)\mapsto 0$,
    $\omega_{A}$ as $(x,y)\mapsto 1$, and $c_0=0$, we have to search
    for a path of weight at most $q$.
  \item[dependency-treewidth] This parameter is in essence defined via
    graph searching game that is equal to the game we study with
    some fly- and reveal-moves forbidden. Forbidding a move can be
    archived by setting the weight of the corresponding edge to
    $\infty$ and by searching for an edge-alternating path of weight
    less then $\infty$.
  \end{compactdesc}

\begin{proof}
Let us first observe that, by the definition of the colosseum, $k$
searchers in the search game have a winning strategy if, and only if,
the start configuration $V(G)$ is contained in $\ball(Q)$. With other
words, if there is an edge-alternating path from $V(G)$ to some
winning configuration in $Q$. Note that such a path directly
corresponds to the strategy by the searchers in the sense that the
used edges directly correspond to possible actions of the searchers.

Since for any graph $G=(V,E)$ and any number $k\in\mathbb{N}$ the
edge-alternating graph $\colosseum(G,k)$ is universal consistent by
Lemma~\ref{lemma:universal_consistent}, all vertices of an edge
alternating path corresponding to a winning strategy are contained in
$\ball(Q)$ as well. In fact, \emph{every} edge-alternating path from
$V(G)$ to $Q$ (and, thus, \emph{any} winning strategy) is completely
contained in $\ball(Q)$. Therefore, it will always be sufficient to
search such paths within $\pit(G,k)$.
By Lemma~\ref{lemma:distance} we can find such a path in time
$O(|\pit(G,k)|^2)$. In fact, we can even define two weight functions
$w_E\colon E\rightarrow\mathbb{N}$ and $w_A\colon
A\rightarrow\mathbb{N}$ and search a \emph{shortest path} from $V(G)$
to $Q$. To compute the invariants of $G$ as stated in the theorem, we
make the following claim:

\begin{claim}
  Let $G=(V,E)$ be a graph and $k\in\mathbb{N}$. Define $w_E$ as $(x,y)\mapsto 0$
  and $w_A$ as $(x,y)\mapsto 1$, and set $c_0=0$. Then we have $d(V(G),Q)\leq q$ in $\pit(G,k)$
  if, and only if, $\twq(G)\leq k-1$.
\end{claim}

\begin{proof}
We follow the proof of Theorem 1
in \cite{FominFN09} closely. We will use the following well-known fact that
easily follows from the observation that in a tree decomposition $(T,\iota)$,
for each three different nodes $i_{1},i_{2},i_{3}\in T$, we have
$\iota(i_{1})\cap \iota(i_{3})\subseteq \iota(i_{2})$ if $i_{2}$ is on the
unique path from $i_{1}$ to $i_{3}$ in $T$.
\begin{fact}\label{lemma:separator}
  Let $(T,\iota)$ be a tree decomposition of $G=(V,E)$ rooted arbitrarily at
  some node $r\in T$. Let $i\in T$ be a node
  and $j\in T$ be a child of $i$ in $T$. Then, the set $\iota(i)\cap \iota(j)$
  is a separator between $C=\bigl[ \bigcup_{d\in \Desc(j)}\iota(d)
  \bigr]\setminus \big(\iota(i)\cap\iota(j)\big)$ and $\big(V\setminus C\big)\setminus \big(\iota(i)\cap\iota(j)\big)$, where $\Desc(x)$
  denotes the set of descendants of $x$ including $x$. Hence, every path from some node
  $u\in C$ to some node $v\in V\setminus C$ contains a vertex of $\iota(i)\cap
  \iota(j)$. 
\end{fact}

\paragraph{From a Tree Decomposition to an Edge-alternating Path:} Let $(T,\iota)$ be a $q$-branched tree decomposition of $G=(V,E)$ of width $k$.
Without loss of generality, we can assume that $G$ is connected. We will
show how to construct an edge-alternating path from the start configuration
$V$ of cost at most $q$ in $\colosseum(G,k+1)$. As described above, this is also
an edge-alternating path with the same costs in $\pit(G,k+1)$. 
The first existential edge from $V$ leads to the configuration $V\setminus \iota(r)$, where
$r$ is the root of $T$. Clearly, $N(V\setminus \iota(r))\subseteq \iota(r)$.  Now suppose that we have reached a configuration
$C$ with $N(C)\subseteq \iota(i)\in V(\colosseum(G,k+1))$ for some node $i\in T$ and we have
\begin{align*}
 C\subseteq \bigl[ \bigcup_{j\in \Desc(i)}\iota(j) \bigr]\setminus \iota(i) 
\end{align*}
where $\Desc(i)$ are the descendants of $i$ in $T$. Clearly, for $i=r$, this
assumption holds trivially. If $i$ is a leaf in $T$, there are no more
descendants and thus $C=\emptyset$. Hence, have reached a winning configuration
in $\colosseum(G,k+1)$. Therefore, suppose that $i$ is a
non-leaf node. We distinguish two cases:
\begin{itemize}
\item If $i$ has exactly one child $j$, we can find a path $P_{1}$ of
  existential edges leading from $C$ to a configuration $C_{1}$ with
  $N(C_{1})\subseteq \iota(i)\cap \iota(j)$. Moreover, we can also find a path $P_{2}$ of
  existential edges from $C_{1}$ to a configuration $C_{2}$ with
  $N(C_{2})\subseteq \iota(j)$.

  The path $P_{1}$ will be constructed by iteratively removing all vertices
  $v\in C$ with $N(v)\cap [\iota(i)\setminus \iota(j)]\neq \emptyset$. For the
  remaining vertices $C_{1}$, we have $N(C_{1})\subseteq \iota(i)\cap \iota(j)$.
  Clearly, if all configurations that we aim to visit on $P_{1}$ exists, the
  corresponding edges also exists by definition. Hence, assume that we are
  currently in some configuration $C'$ with $N(C')\cap [\iota(i)\setminus
  \iota(j)]\neq \emptyset$ and want to remove some vertex $v\in C'$ with
  $N(v)\cap [\iota(i)\setminus \iota(j)]\neq \emptyset$, but $C'\setminus
  \{v\}\not\in V(\colosseum(G,k+1))$. By definition of $\colosseum(G,k+1)$, this
  means that $|N(C'\setminus \{v\})|\geq k+2$. As we wanted to remove $v$, we
  have $N(v)\cap \iota(i)\neq \emptyset$. On the other hand, as $N(C'\setminus
  \{v\})\subseteq N(C')\cup \{v\}$ and $|N(C'\setminus \{v\})|\geq k+2$, we know
  that there is some $u\in C'$ with $v\in N(u)$. Hence,
  Fact~\ref{lemma:separator} implies that $v\in \iota(i)\cap \iota(j)$, a
  contradiction. Hence, all configurations in $P_{1}$ exist. 

  Similarly, we construct $P_{2}$ by iteratively removing all vertices in
  $\iota(j)$ from  $C_{1}$. It is easy to see that the neighborhood of the
  visited configurations will always be a subset of $\iota(j)$ and hence, all
  configurations on this path exist. 

We have thus arrived at a configuration $C_{2}$ with $N(C_{2})\subseteq
\iota(j)$ and
\begin{align*}
C_{2}\subseteq  \bigl[ \bigcup_{j'\in \Desc(j)}\iota(j') \bigr]\setminus \iota(j)
\end{align*}
due to Fact~\ref{lemma:separator}. 

\item If node $i$ has a set of children $J$ with $|J|\geq 2$, we will use
  universal edges.  Let $\mathcal{C}$ be
  the connected components of $G[\bigcup_{j\in \Desc(i)}
  \iota(j)\setminus\iota(i)]$.  
  We claim, that for each component $\Gamma\in \mathcal{C}$, there is a
  unique index $j(\Gamma)\in J$ such that $\Gamma\cap \iota(j(\Gamma))\neq \emptyset$. If no such
  index exists, we have $\iota(j)=\iota(i)$. We can iteratively remove such bags
  $\iota(j)$ until this can not happen anymore. If two indices $j_{1},j_{2}\in
  J$ exist with $\iota(j_{1})\cap \Gamma\neq \emptyset$ and $\iota(j_{2})\cap \Gamma\neq
  \emptyset$, the connectivity property implies that $\iota(i)\cap \Gamma\neq
  \emptyset$, a contradiction to our assumption. Hence, for each component $\Gamma$,
  we follow the universal edge to $\Gamma$ and then proceed as above: first, we find
  a path $P_{1}$ of existential edges from  $\Gamma$ to a configuration
  $\Gamma_{1}$ with $N(\Gamma_{1})\subseteq \iota(i)\cap \iota(j(\Gamma))$ and
  then a path $P_{2}$ of existential edges from $\Gamma_{1}$ to a configuration
  $\Gamma_{2}$ with $N(\Gamma_{2})\subseteq \iota(j(\Gamma))$. The same
  arguments as above imply that all configurations on these paths exist and that
  we arrive at a configuration $\Gamma_{2}$ with $N(\Gamma_{2})\subseteq
  \iota(j(\Gamma))$ and 
  \begin{align*}
\Gamma_{2}\subseteq \bigl[ \bigcup_{j'\in \Desc(j(\Gamma))}\iota(j') \bigr]\setminus \iota(j(\Gamma)).
\end{align*}

\end{itemize}
This shows that we will eventually reach the leaves of the tree decomposition
and thus some wining configuration. Clearly, this is an edge-alternating path in
$\colosseum(G,k+1)$ and thus in $\pit(G,k+1)$. Furthermore, as each path from
the root of $T$ to some leaf of $T$ contains at most $q$ nodes with more than
one children, this path is $q$-branched, as we use at most $q$ universal edges
from the initial configuration $V$ to any used winning configuration for every
induced directed path. Hence, we have found an edge-alternating path in
$\pit(G,k+1)$ of cost at most $q$.

\paragraph{From an Edge-alternating Path to a Tree Decompostion:} Let
$P\subseteq V(\pit(G,k+1))$ be an edge-alternating $q$-branched path from the initial
configuration $V$ to a final configuration $\{v^{*}\}$ in $\pit(G,k+1)$ with
$|N(\{v^{*}\})| \leq k$. We argue
inductively on~$q$.
\begin{itemize}
\item If $q=0$, the path $P$ does not use any universal edges. Let
  $\pi=\pi_{1},\ldots,\pi_{s}$ be any classical directed path from the initial
  configuration $V$ to some wining configuration $\{v^{*}\}$ in $\pit(G,k+1)$ that only uses
  vertices from $P$. As the initial configuration is $\pi_{1}=V$, the winning
  configuration is $\pi_{s}=\{v^{*}\}$, and there are only
  existential edges $(C,C')$ with $|C'|=|C|-1$ in $\pit(G,k+1)$, we know that
  $|\pi_{i}|=|V|-i+1$ and thus $s=|V|$. We say that vertex $v\in V$ is
  \emph{removed at time $i$}, if $v\in \bigcap_{j=1}^{i}\pi_{j}$ and $v\not\in
  \bigcup_{j=i+1}^{|V|}\pi_{j}$. We also say that $v^{*}$ was removed at time $|V|$.
  For $i=1,\ldots,|V|$, let $v_{i}$ be the vertex
  removed at time $i$.

  We will now construct a $0$-branched tree decomposition
  $(T,\iota)$, i.\,e.~a path decomposition. As $T$ is a path, let
  $t_{1},\ldots,t_{|V|}$ be the vertices on the path in their respective ordering
  with root $t_{1}$. We set $\iota(t_{i})=N(\pi_{i})\cup \{v_{i}\}$. For $i=1,\ldots,|V|-1$, there is
  an existential edge leading from $\pi_{i}$ to $\pi_{i+1}$ and thus
  $|N(\pi_{i})|\leq k$. As $\pi_{|V|}=\{v_{|V|}\}$ is a winning
  configuration, we also have $|N(\pi_{|V|})|\leq k$. Hence, the 
  resulting decomposition $T$ has width at most $k$. As $T$ is a path, it is
  also $0$-branched.

  We now
  need to verify that $(T,\iota)$ is indeed a valid tree decomposition. As every
  vertex $v$ is removed at some time $i$, we have $v=v_{i}$ and thus $v\in
  \iota(t_{i})$. Hence, every vertex is in some bag. Let $\{v_{i},v_{i'}\}$ be any edge with $i < i'$. As $v_{i'}\in
  \pi_{i'}$ and $v_{i}\not\in \pi_{i'}$, we have $v_{i}\in N(\pi_{i'})$ and thus
  $\{v_{i},v_{i'}\}\subseteq N(\pi_{i'})\cup \{v_{i'}\} = \iota(t_{i'})$. Hence,
  every edge is in some bag. Finally, let $v_{i}\in V$. Clearly, as $v_{i}\in
  \pi_{1}$, $v_{i}\in \pi_{2}$,\ldots, $v_{i}\in \pi_{i-1}$, the first bag where
  $v_{i}$ might appear is $\iota(t_{i})$. Let $v_{i'}\in N(v_{i})$ be the
  neighbour of $v_{i}$ that is removed at the latest time. If $i' < i$, we have
  $N(v_{i})\cap \bigcup_{j=i+1}^{|V|}\pi_{j} = \emptyset$ and $v_{i}$ thus only
  appears in $\iota(t_{i})$. If $i < i'$, then $v_{i}\in
  \bigcap_{j=i+1}^{i'}N(\pi_{j})$ and hence $v_{i}\in
  \bigcap_{j=i+1}^{i'}\iota(t_{j})$. 

\item Now, assume that $q\geq 1$ and that we can construct for every $q' < q$ a
  $q'$-branched tree decomposition of width at most $k$ from any $q'$-branched
  edge-alternating path $P$ in $\pit(G,k+1)$. Consider the directed acyclic
  subgraph $H$ in $\pit(G,k+1)$ induced by $P$. A configuration $C\in V(H)$ is
  called a \emph{universal configuration}, if $N_{A}(C)\subseteq V(H)$ and a
  \emph{top-level universal configuration} with respect to some directed path
  $\pi$ if $C$ is the first universal configuration on $\pi$. 
  Note that we can reduce $P$ in such a way that all directed paths $\pi$
  from the initial configuration $V$ to some winning configuration $\{v^{*}\}$
  in $H$
  have the same top-level universal configuration, call it $C^{*}$. Let
  $V=\pi_{1},\ldots,\pi_{i}=C^{*}$ be the shared existential path from $V$ to
  $C^{*}$ in $H$ and let $N_{A}(C^{*})=\{C_{1},\ldots,C_{\ell}\}$
  be the universal children of $C^{*}$. Note that
  $\{C_{1},\ldots,C_{\ell}\}\subseteq P$ due to the definition of an
  edge-alternating path. For each child $C_{j}$, the edge-alternating path $P$
  contains a directed path $\pi^{(j)}$ from $C_{j}$ to some final configuration
  in $\pit(G,k+1)$. Furthermore, each $\pi^{(j)}$ contains at most $q'\leq q-1$
  universal edges (otherwise, $P$ would not be $q$-branched). Hence, by
  induction hypothesis, we can construct a $q'$-branched tree decomposition
  $(T^{(j), \iota^{(j)}})$ for the subgraph induced by the vertices contained in
  the path $\pi^{(j)}$ with root $r^{(j)}$. 
  
  Now, we use the same construction as above to construct a path
  $(T'=(t'_{1},\ldots,t'_{i}),\iota')$ from $\pi_{1},\ldots,\pi_{i}$ and for
  each path $\pi^{(j)}$, we add the root $r^{(j)}$ of the $q'$-branched tree decomposition
  $(T^{(j)},\iota(j))$ as a child of bag $t_{i}$ to obtain our final tree
  decomposition $(T,\iota)$. As there is a universal edge
  from $C^{*}$ to $C_{j}$, we know that $C_{j}$ is a component of
  $C^{*}$. As all $(T^{(j)}, \iota(j))$ are valid $q-1$-branched tree
  decompositions of width at most $k$, we can thus conclude that $(T,\iota)$ is a
  valid $q$-branched tree decomposition of width $k$.\qedhere
\end{itemize}
\end{proof}

Combining the above claim with Theorem~\ref{theorem:pid} for computing
the pit, we conclude that we can check whether a graph $G$ has
$q$-branched-treewidth $k$ in time $O(|\pit(G,k+1)|^2\cdot|V|^2)$. We
note that the algorithm is fully constructive, as the obtained path
(and, hence, the winning strategy of the searchers) directly
corresponds to the desired decomposition.
Since we have $\tw(G)=\tw_{\infty}(G)$ and $\pw(G)=\tw_0(G)$, the
above results immediately implies the same statement for treewidth and
pathwidth by checking $d(V(G),k)<\infty$  or $d(V(G),k)=0$, respectively.

In order to show the statement for treedepth, we will require another
claim for different weight functions. The proof idea is, however, very similar.

\begin{claim}
  Let $G=(V,E)$ be a graph and $k\in\mathbb{N}$. Define $w_E$ as $(x,y)\mapsto 1$
  and $w_A$ as $(x,y)\mapsto 0$, and $c_0=1$. Then we have
  $d(V(G),Q)\leq k$ in $\pit(G,k)$ if, and only if, $\td(G)\leq k$.
\end{claim}

\begin{proof}
  To prove the claim, we use an alternative representation of
  treedepth~\cite{NesetrilM2012}. Let $G=(V,E)$ be a graph with
  connected components $C_1,\dots,C_{\ell}$, then:
  \[
    \td(G)=\begin{cases}
      1 & \text{if $|V|=1$;}\\
      \max_{i=1}^{\ell}\td(G[C_i]) & \text{if $\ell\geq 2$;}\\
      \min_{v\in V}\td(G[V\setminus\{v\}])+1 & \text{otherwise.}
    \end{cases}
  \]
  Let us reformulate this definition a bit. Let $C\subseteq V$ be a
  subset of the vertices and let $C_1,\dots,C_{\ell}$ be the connected
  components of $G[C]$. Define:
  \[
    \td^*(C)=\begin{cases}
      1 & \text{if $|C|=1$;}\\
      \max_{i=1}^{\ell}\td^*(C_i) & \text{if $\ell\geq 2$;}\\
      \min_{v\in C}\td^*(C\setminus\{v\})+1 & \text{otherwise.}
    \end{cases}
  \]
  Obviously, $\td(G)=\td^*(V)$. We proof that for any $C\subseteq V$
  we have $d(C,Q)=\td^*(C)$ in $\pit(G,k)$ for every $k\geq\td(G)$ and
  $d(C,Q)\geq\td^*(C)$ for all $k<\td(G)$.

  For the first part we consider the vertices of $\pit(G,k)$ in
  inverse topological order and prove the claim by induction. The
  first vertex $C_0$ is in $Q$ and thus $d(C_0,Q)=c_0=1$. Since the
  vertices in $Q$ represent sets of cardinality $1$, we have
  $d(C_0,Q)=\td^*(C_0)$. For the inductive step consider $C_i$ and
  first assume it is not connected in $G$. Then
  \begin{align*}
    d(C_i,Q)&=\max_{C_j\in N_{\forall}(C_i)}\big(d(C_j,Q)+w_A(C_i,C_j)\big)\\
            &=\max_{C_j\in N_{\forall}(C_i)}d(C_j,Q)\\
            &=\max_{\text{$C_j$ is a component in $G[C_i]$}}\td^*(C_j)\\
            &=\td^*(C_i).
  \end{align*}
  Note that there could, of course, also be existential edges leaving
  $C_i$. However, since the universal edges are ``for free,'' for every
  shortest path that uses an existential edge at $C_i$, there is also
  one that first uses the universal edges.

  For the second case, that is $C_i$ is connected, observe that $C_i$
  is not incident to any universal edge. Therefore we obtain:
  \begin{align*}
    d(C_i,Q)&=\min_{v\in C_i}\big(d(C_i\setminus\{v\},Q)+w_E(C_i,C_i\setminus\{v\})\big)\\
            &=\min_{v\in C_i}\big(d(C_i\setminus\{v\},Q)+1\big)\\
            &=\min_{v\in C_i}\big(\td^*(C_i\setminus\{v\})+1\big)\\
            &=\td^*(C_i).
  \end{align*}
  This completes the part of the proof that shows $d(C,Q)=\td^*(C)$
  for $k\geq\td(G)$.
  We are left with the task to argue that $d(C,Q)\geq\td^*(C)$ for all
  $k<\td(G)$. This follows by the fact that for every $k'<k$ we have
  that $\pit(G,k')$ is an induced  subgraph of $\pit(G,k)$. Therefore,
  the distance can only increase in the pit for a smaller $k$~--~in fact,
  the distance can even become infinity if $k<\td(G)$.
\end{proof}

Again, combining the claim with Theorem~\ref{theorem:pid} yields the
statement of the theorem for treedepth.
Finally, we will prove the statement for dependency-treewidth. This
parameter can be characterized by a small adaption of the graph
searching game~\cite{EibenGO2018}: In addition to the graph $G$ and
the parameter $k$, one is also given a partial ordering $\lessdot$ on
the vertices of $G$. For a vertex set $V'$, let
$\mu_{\lessdot}(V')=\{\,v\in V'\mid \forall w\in V'\setminus\{v\}:
(w,v)\not\in \lessdot\,\}$ be the minimal elements of $V'$ with regard
to $\lessdot$. If $C\subseteq V(G)$ is the contaminated area, we are
only allowed to put a searcher on $\mu_{\lessdot}(C)$, rather than on
all of $C$. The \emph{dependency-treewidth} $\dtw_{\lessdot}(G)$ is
the minimal number of searchers required to catch the fugitive in this
version of the game. Therefore, we just need a way to permit only
existential edges $(C,C')$ with $C\setminus C'\subseteq
\mu_{\lessdot}(C)$. We show the following stronger claim:
\begin{claim}
  Consider a variant of the search game in which at some
  configurations $C_i$ some fly-moves are forbidden, and in which
  furthermore at some configurations $C_j$ no reveals are allowed. Whether
  $k$ searcher have a winning strategy in this game can be decided in
  time $O(|\pit(G,k)|^2\cdot|V|^2)$.
\end{claim}
\begin{proof}
  First observe that, if the $k$ searcher have a winning strategy $S$,
  this strategy corresponds to a path in $\pit(G,k)$. The reason is
  that searchers that are allowed to use all fly- and reveal-moves
  (and for which all winning strategies correspond to paths in
  $\pit(G,k)$) can, of course, use $S$ as well. We compute the pit
  with Theorem~\ref{theorem:pid}.

  Now to find the restricted winning strategy we initially set $w_E$
  and $w_A$ to $(x,y)\mapsto 0$. Then for any existential edge
  $(C_i,C_j)$ that we wish to forbid we set
  $w_E(C_i,C_j)=\infty$. Furthermore, for any node $C$ at witch we
  would like to forbid universal edges we set $w_E(C_i,C_j)=\infty$
  for all $C_j\in N_{\forall}(C_i)$. Finally, we search a path from
  $V(G)$ to $Q$ of weight less then $\infty$ using Lemma~\ref{lemma:distance}.
\end{proof}
This completes the proof of Theorem~\ref{theorem:parameters}.
\end{proof}

\clearpage
\section{Theoretical Bounds for Certain Graph Classes}
In general, it is hard to compare the size of the arena, the colosseum, and the
pit. For instance, already simple graph classes as paths ($P_{n}$) and stars
($S_{n}$) reveal that the colosseum may be smaller or larger than the
arena (the arena has size $O(n^3)$ on both, but the colosseum has size
$O(n)$ on $P_n$ and $O(2^n)$ on $S_n$, both with regard to their optimal
treewidth $1$).
However, experimental data of the PACE challenge~\cite{PaceIpec16,PaceIpec17} shows that
the pit is very small in practice. In the following, we are thus interested in
graph classes where we can give theoretical guarantees on the size of the pit.
We will first show that the colosseum is indeed often smaller than the arena
(Lemma~\ref{lemma:claw-free}) and furthermore, that the pit might be much
smaller than the colosseum (Lemma~\ref{lemma:prop}).

\begin{lemma}\label{lemma:claw-free}
  For every connected claw-free graph $G=(V,E)$ and integer $k\in\mathbb{N}$, it holds
  that  $|\colosseum(G,k)|\leq
  \sum_{i=1}^k\binom{n}{i}\cdot 2^{2i}\in O(\binom{n}{k}\cdot 4^k)$.
 \end{lemma}
\begin{proof}
  Observe that in a claw-free graph every $X\subseteq V$ separates $G$ in
  at most $2\cdot|X|$ components, as every component is connected to a vertex in
  $X$ (since $G$ is connected), but every vertex in $X$ may be connected to at
  most two components (otherwise it forms a claw). In the colosseum, every
  configuration $C$ corresponds to a separator $N(C)$ of size at most $k$, and
  there are at most $\sum_{i=1}^k\binom{n}{i}$ such separators. For each
  separator we may combine its associated components in an arbitrary fashion to
  build configurations of the colosseum, but since there are at most $2\cdot i$
  components, we can build at most $2^{2\cdot i}$ configurations.
\end{proof}
We remark that the result of Lemma~\ref{lemma:claw-free} can easily be
extended to $K_{1,t}$-free graphs for every fixed $t$, and that this
result is rather tight:
\begin{lemma}\label{lemma:col_large}
Let $G=(V,E)$ be a graph and $k\in \mathbb{N}$. It holds that
$|\colosseum(G,k)|\geq \sum_{i=1}^{k}\binom{|V_{i}|}{i}$, where $V_{i}=\{v\in V
:  |N(v)|\geq i\}$. 
\end{lemma}
\begin{proof}
  Let $X$ be any subset of at most $i$ vertices from $V_{i}$ with
  $i\leq k$. As $|X|\leq i$, every vertex in $X$ has a neighbour in $V\setminus X$. Hence, $N(V\setminus
  X)=X$ and thus $|N(V\setminus X)|\leq k$ and $V\setminus X\in V(\colosseum(G,k))$.
\end{proof}
We now show that the pit, on the other hand, can be substantially smaller than the colosseum even
for graphs with many high-degree vertices. 
For $n,k\in \mathbb{N}$ with $n\geq 2k$, we define the graph $P_{n,k}$ on vertices
$V(P_{n,k})=\{v_{0},v_{1},\ldots,v_{n\cdot k},v_{n\cdot k+1}\}$. For $i=1,\ldots,n$, let
$X_{i}=\{v_{(i-1)\cdot k+1},v_{(i-1)\cdot k+2},\ldots, v_{(i-1)\cdot k+k}\}$, $X_{0}=\{v_{0}\}$, and
$X_{n+1}=\{v_{n\cdot k+1}\}$. The edges $E(P_{n,k})$ are defined as
\begin{align*}
  E(P_{n,k})=\bigcup_{i=1}^{n}\{\{u,v\} \mid u,v\in X_{i}\}\cup \bigcup_{i=0}^{n}\{\{u,v\}\mid u\in X_{i}, v\in X_{i+1}\}.
\end{align*}
Informally, $P_{n,k}$ is constructed by taking a path of length $n+2$ and
replacing the inner vertices by cliques of size $k$ that are completely
connected to each other.

\begin{lemma}\label{lemma:prop}
  It holds:
{\begin{compactenum}[(i)]
  \item $\tw(P_{n,k})=\pw(P_{n,k})=2k-1$;
    \hfill  (ii) $|\arena(P_{n,k},2k)|=2\cdot\binom{n\cdot k+2}{2k+1}$;
    \setcounter{enumi}{2}
  \item $|\colosseum(P_{n,k},2k)|\geq \sum_{i=1}^{2k}\binom{n\cdot k}{i}$;
    \hfill  (iv) $|\pit(P_{n,k},2k)|\in O(n^2+n\cdot 2^{6k})$.
  \end{compactenum}}
\end{lemma}
\begin{proof}
   Property~\textit{(i)} holds as $P_{n,k}$ contains a clique of
  size $2k$ and the path decomposition $[X_{0}\cup X_{1}, X_{1}\cup X_{2},
  \ldots, X_{n}\cup X_{n+1}]$ is valid. Property~\textit{(ii)}
  holds by definition of $\arena(P_{n,k},2k)$. To see~\textit{(iii)}, observe that for $v\in \bigcup_{i=1}^{n}X_{i}$, we
  have $|N(v)|\geq 2k$. Lemma~\ref{lemma:col_large} thus implies
  $|\colosseum(P_{n,k},2k)|\geq \sum_{i=1}^{2k}\binom{n\cdot k}{i}$.

  In order to prove Property~\textit{(iv)} we count the number of
  configurations inserted into the queue by algorithm \textsf{Discover}.
  Theorem~\ref{theorem:pid} shows that this number equals the size of the pit.
  All configurations either include $v_{0}$ (a \emph{left} configuration),
  $v_{k\cdot n+1}$ (a \emph{right} configuration) or both (a \emph{mixed}
  configuration), as $\{v_0\},\{v_{k\cdot n+1}\}$ are the only winning
  configurations. Left and right configurations can be extended by
  reverse fly-moves (Line~\ref{discover:F}) as follows: Starting from $C=X_{0}$, we can add a vertex
  $u\in X_1$ to $C$ to generate $C'$. The neighborhood of $C'$ will be
  $\big(X_1\setminus\{u\}\big)\cup X_2$. Adding a vertex $w$ of $X_2$ to $C'$
  would result in a configuration with neighborhood
  $\big(X_1\setminus\{u\}\big)\cup\big(X_2\setminus\{w\}\big)\cup X_3$ greater
  than $2k$. Hence, further reverse fly-moves have to add $X_1$ completely to
  $C'$ before elements of $X_2$ can be added. An inductive arguments yields that
  configurations constructed in this way have the form $X_0\cup X_1\cup\dots\cup
  X_{i-1}\cup\tilde X_i$ with $\tilde X_i\subseteq X_i$. As the same is true
  when starting with $X_{n+1}$, we can generate $2\cdot n\cdot 2^k$ such configurations.

  Now consider reverse reveal-moves (Line~\ref{discover:R}). We can only unite a left
  configuration $C_1=X_0\cup X_1\cup\dots\cup X_{i-1}\cup \tilde X_i$
  with a right one $C_2=X_{n+1}\cup X_n\cup\dots\cup X_{j+1}\cup \tilde X_j$.
  If $\tilde{X_{i}}=\tilde{X_{j}}=\emptyset$, $C_1\cup C_2$ is a legal
  configuration and there are $\binom{n}{2}$ such configurations. 
  If $|\tilde{X_{i}}\cup \tilde{X_{j}}| > 0$, we have $|N(C_{1})|+|N(C_{2})| >
  2k$, and can not add such configurations. For the combinations with $i + 1 < j - 2 $, we have $C_{j-2} \cap
  (N(C_{1})\cup N(C_{2}))=\emptyset$ and thus $N(C_{1})\cap N(C_{2}) =
  \emptyset$, which implies that $C_{1}\cup C_{2}$ is not legal
  due to $|N(C_{1})\cup N(C_{2})|=|N(C_{1})|+|N(C_{2})|>2k$. The remaining
  combinations have $i+1\geq j-2$. For each fixed $i$, there are at most four
  such values of $j$ ($j\in \{i,i+1,i+2,i+3\}$). Both $\tilde{X_{i}}$ and
  $\tilde{X_{j}}$ might be arbitrary and we can thus create at most $4\cdot
  n\cdot 2^{k}\cdot 2^{k}=4\cdot n\cdot 2^{2k}$ configurations.

  Finally, we can perform reverse fly-moves for mixed configurations
  of the form $C=X_0\cup
  X_1\cup\dots\cup X_{i-1}\cup \tilde X_i \cup \tilde{X_{j}}\cup X_{j+1} \cup
  \dots \cup X_{n}\cup X_{n+1}$ with $i+1\geq j-2$ (otherwise, the neighborhood
  is too large). Let $U=V\setminus C$ be the uncontaminated vertices. As
  $U\subseteq X_{i}\setminus \tilde{X_{i}}\cup X_{i+1}\cup X_{i+2}\cup
  X_{j}\setminus \tilde{X_{j}}$, we have $|U|\leq 4k$. Hence, for each of the
  $4\cdot n\cdot 2^{2k}$ configurations with $i+1\geq j-2$ there are at most
  $2^{4k}$ configurations reachable by reverse fly-moves~--~yielding at most
  $4\cdot n\cdot 2^{6k}$ configurations. Overall, we thus have
  $n\cdot 2^{k}+\binom{n}{2}+4\cdot n\cdot 2^{6k}$ configurations.
\end{proof}

\section{Experimental Estimation of the Pit Size}
A heavily optimized version of the treewidth algorithm described above has been
implemented in the Java library Jdrasil~\cite{BannachBE17,BannachBE17J}. To
show the usefulness of our general approach, we experimentally compared the size
of the pit, the arena, and the colosseum for various named graphs known from the
DIMACS Coloring Challenge~\cite{dimacsColoring} or the
PACE~\cite{PaceIpec16,PaceIpec17}. For each graph the values are taken for the minimal
$k$ such that $k$ searchers can win. Note that
$|\arena(G,k)|\leq |\pit(G,k)|$ holds only in $6$ of $24$ cases, 
emphasized by underlining.

\resizebox{0.475\textwidth}{!}{\scriptsize%
  \begin{tabular}[t]{lcccrrr}
    \toprule
    Graph & $|V|$ & $|E|$ & $k$ & Pit & Arena & Col. \\
    \cmidrule(rl){1-7}
    Grotzsch & 11 & 20 & 6 & 1,235 & \underline{660} & 1,853\\
    Heawood & 14 & 21 & 6 & 5,601 & 6,864 & 9,984\\
    Chvatal & 12 & 24 & 7 & 3,170 & \underline{990} & 3,895\\
    Goldner Harary & 11 & 27 & 4 & 103 & 924 & 639\\
    Sierpinski Gasket & 15 & 27 & 4 & 488 & 6,006 & 2,494\\
    Blanusa 2. Snark & 18 & 27 & 5 & 861 & 37,128 & 15,413\\
    Icosahedral & 12 & 30 & 7 & 2,380 & \underline{990} & 3,575\\
    Pappus & 18 & 27 & 7 & 54,004 & 87,516 & 97,970\\
    Desargues & 20 & 30 & 7 & 85,146 & 251,940 & 202,661\\
    Dodecahedral & 20 & 30 & 7 & 112,924 & 251,940 & 207,165\\
    Flower Snark & 20 & 30 & 7 & 79,842 & 251,940 & 203,473\\
    Gen. Petersen & 20 & 30 & 7 & 78,384 & 251,940 & 202,685\\
    \bottomrule
  \end{tabular}}
\resizebox{0.475\textwidth}{!}{\scriptsize%
  \begin{tabular}[t]{lcccrrr}
    \toprule
    Graph & $|V|$ & $|E|$ & $k$ & Pit & Arena & Col. \\
    \cmidrule(rl){1-7}
    Hoffman & 16 & 32 & 7 & 5,851 & 25,740 & 30,270\\
    Friendship 10 & 21 & 30 & 3 & 57,554 & 11,970 & 58,695\\
    Poussin & 15 & 39 & 7 & 3,745 & 12,870 & 17,358\\
    Markstroem & 24 & 36 & 5 & 13,846 & 269,192 & 71,604\\
    McGee & 24 & 36 & 8 & 487,883 & 2,615,008 & 1,905,241\\
    Naru & 24 & 36 & 7 & 41,623 & 1,470,942 & 708,044\\
    Clebsch & 16 & 40 & 9 & 20,035 & \underline{16,016} & 55,040\\
    Folkman & 20 & 40 & 7 & 21,661 & 251,940 & 151,791\\
    Errera & 17 & 45 & 7 & 3,527 & 48,620 & 42,418\\
    Shrikhande & 16 & 48 & 10 & 50,627 & 8,736 & 61,456\\
    Paley & 17 & 68 & 12 & 114,479 & \underline{4,760} & 129,474\\
    Goethals Seidel & 16 & 72 & 12 & 54,833 & \underline{1,120} & 65296\vspace{1.1mm}\\
    \bottomrule
  \end{tabular}}
\smallskip

\noindent We have performed the same experiment on various random graph models. For
each model we picked $25$ graphs at random and build the mean over all
instances, where each instance contributed values for its minimal
$k$. We used all 3 models with $N=25$ and, for the first two with
$p=0.33$; and for the later two with $K=5$. For a detailed description
of the models see for instance~\cite{bollobas1998random}.
\begin{center} 
  {\scriptsize
  \begin{tabular}{lccc}
    \toprule
    Model & $|\pit(G,\mathrm{OPT})|$ & $|\arena(G,\mathrm{OPT})|$ & $|\colosseum(G,\mathrm{OPT})|$\\
    \cmidrule(rl){1-4}
    Erd\H{o}s–R\'{e}nyi & 66,320 & 342,918 & 503,767 \\
    Watts Strogats & 15,323 & 192,185 & 108,074 \\
    Barab\'{a}si Albert & 61,147 & 352,716 & 551,661 \\
    \bottomrule
  \end{tabular}}
\end{center}
Finally, we observe the growth of the \textcolor{unired}{pit}, the \textcolor{unigreen}{arena}, and the
\textcolor{uniblue}{colosseum} for a fixed graph if we raise $k$ from $2$ to the optimal
value. While the arena shows its binomial behavior, the
colosseum is in many early stages actually smaller then
the arena. This effect is even more extreme for the pit, which is
\emph{very} small for $k$ that are smaller then the optimum. This
makes the technique especially well suited to establish lower bounds,
an observation also made by Tamaki~\cite{Tamaki2017}.

\begin{center}
  \includegraphics{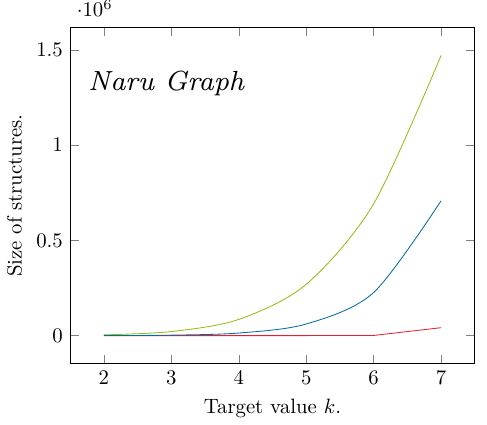}\qquad
  \includegraphics{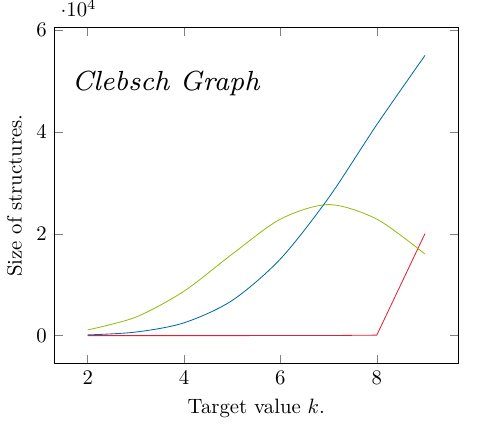}
  
  \includegraphics{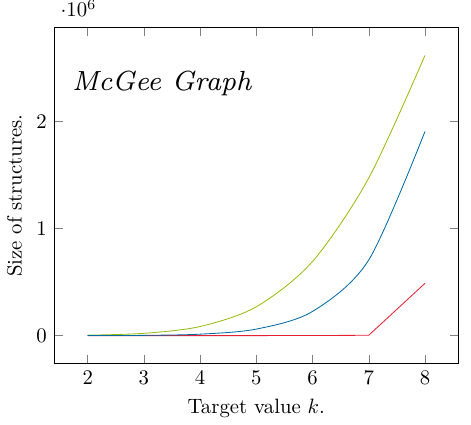}\qquad
  \includegraphics{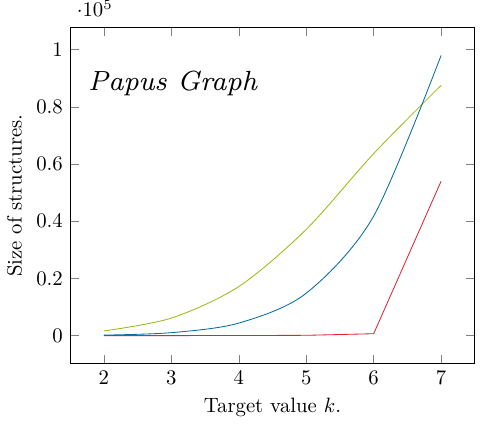}
\end{center}

\section{Conclusion and Outlook}
Treewidth is one of the most useful graph parameters that is
successfully used in many different areas. The Positive-Instance
Driven algorithm of Tamaki has led to the first practically relevant
algorithm for this parameter.  We have formalized Tamaki's algorithm
in the more general setting of graph searching, which has allowed us
to (i) provide a clean and simple formulation; and (ii) extend the
algorithm to many natural graph parameters.  With a few further
modification of the colosseum, our approach can also be used for the
notion of \emph{special-treewidth}~\cite{Courcelle12}. We assume that
a similar modification may also be possible for other parameters such
as \emph{spaghetti-treewidth}~\cite{BodlaenderKKKO17}.

\paragraph{Acknowledgements:}The authors would like to thank Jan Arne Telle and Fedor Fomin for helpful
discussions about the topic and its presentation.

\end{document}